\documentclass[%
 aip,
 jmp,%
amsmath,amssymb,
preprint,%
]{revtex4}

\usepackage{mathtools}
\usepackage{amsthm}
\usepackage{dsfont}
\usepackage{amstext}
\usepackage{esint}
\usepackage{bm,bbm}  
\usepackage{eucal}
\usepackage{verbatim}
\usepackage{booktabs,array,multirow}
\usepackage{times}
\usepackage{graphicx,color}
\usepackage{tabularx}
\usepackage{dsfont}
\usepackage{titletoc} 
\usepackage{alltt}
\usepackage{microtype}
\usepackage[colorlinks=true,linkcolor=black,citecolor=black,urlcolor=blue]{hyperref}
\usepackage{graphicx}
\usepackage{dcolumn}
\usepackage{boxedminipage,changebar}

\newcolumntype{C}{>{$}c<{$}}
\AtBeginDocument{
\heavyrulewidth=.08em
\lightrulewidth=.05em
\cmidrulewidth=.03em
\belowrulesep=.65ex
\belowbottomsep=0pt
\aboverulesep=.4ex
\abovetopsep=0pt
\cmidrulesep=\doublerulesep
\cmidrulekern=.5em
\defaultaddspace=.5em
}

\theoremstyle{plain}
\newtheorem{theorem}{Theorem}[section]

\newtheorem{lemma}[theorem]{Lemma}
\newtheorem{proposition}[theorem]{Proposition}
\newtheorem{cor}[theorem]{Corollary}

\newtheorem{conj}{Conjecture}

\theoremstyle{remark}
\newtheorem{rmk}{Remark}

\newcommand{\barr}{\begin{eqnarray}}
\newcommand{\earr}{\end{eqnarray}}
\newcommand{\be}{\begin{equation}}
\newcommand{\ee}{\end{equation}}

\newcommand{\numberset}{\mathbb}
\newcommand{\N}{\numberset{N}}
\newcommand{\Z}{\numberset{Z}}
\newcommand{\R}{\numberset{R}}

\newcommand{\tr}{\mathrm{Tr}}
\newcommand{\Tr}{\mathrm{Tr}}

\def\im{{\rm i}}

\newcommand{\E}{\mathbb{E}}

\begin{document}


\title
{Large-$N$ expansion for the time-delay matrix of ballistic chaotic cavities}

\author{Fabio Deelan Cunden}
\email{fabiodeelan.cunden@bristol.ac.uk}
 \affiliation{School of Mathematics, University of Bristol, University Walk, Bristol BS8 1TW, United Kingdom}

\author{Francesco Mezzadri}
\affiliation{School of Mathematics, University of Bristol, University Walk, Bristol BS8 1TW, United Kingdom}%

\author{Nick Simm}
\affiliation{Mathematics Institute, University of Warwick, Coventry CV4 7AL, United Kingdom}%

\author{Pierpaolo Vivo}
\affiliation{King's College London, Department of Mathematics, Strand, London WC2R 2LS, United Kingdom}%

\date{\today}

\begin{abstract}
We consider the $1/N$-expansion of the moments of the proper delay times for a ballistic chaotic cavity supporting $N$ scattering channels. In the random matrix approach, these moments correspond to traces of \emph{negative} powers of Wishart matrices. For systems with and without broken time reversal symmetry (Dyson indices $\beta=1$ and $\beta=2$) we obtain a recursion relation, which efficiently generates the coefficients of the $1/N$-expansion of the moments. The integrality of these coefficients and their possible diagrammatic interpretation is discussed. 
\end{abstract}

%
\maketitle
\tableofcontents

\section{Introduction and statement of the results}
\subsection{Background}
The Wigner-Smith~\cite{Smi60,Wig55,Eis48} time-delay matrix $Q$ plays a central role in the theory of quantum transport \cite{Fyodorov97,Texier15}. It is defined in terms of the $N$-channel
scattering matrix $S$ via the relation
\begin{equation}
  \label{eq:wdtm}
  Q = -\im\hbar S^{\dagger}(E)\frac{\partial S(E)}{\partial E}\ ,
\end{equation}
where $E$ is the energy of the incoming particle.
If $S$ is unitary, it is easily seen that $Q$ is Hermitian. The eigenvalues $\tau_1,\dots,\tau_N$ of $Q$ are called
\emph{proper delay times}. Apart from the scalar case ($N=1$), the individual proper delay times $\tau_k$ have no immediate physical meaning. Physically relevant quantities are instead unitarily invariant functions of $Q$ such as powers of traces $\Tr Q^k=\tau_1^k+\cdots+\tau_N^k$.  In fact, several measurable observables are entirely determined by these powers of traces. For instance, the Wigner delay time $\Tr Q$ is a bona fide measure of the time spent by an incident particle in the scattering region. Higher powers turn out to play a role in AC electronic transport~\cite{ButtThomPre93,GopMelBut96,ButPol05}~, e.g. in the low frequency expansion ($\omega\to0$) of the AC dimensionless conductance \mbox{$G(\omega)=[-\im\omega\Tr Q+(1/2)\omega^2\Tr Q^2+...]$}. 

Using the random matrix approach to ballistic chaotic scattering, with the assumption that the internal Hamiltonian of the cavity belongs to a unitarily invariant ensemble with a large number of number of bound states,
 Brouwer, Frahm and Beenakker~\cite{BroFraBee97_99} showed (in quite an ingenious
way) that the eigenvalues of $(N Q)^{-1}$ are distributed as those of matrices in the Laguerre
ensemble. In other words, denoting the latter \emph{rates} by $\lambda_1,\dots,\lambda_N$, their joint probability density is supported on $\R_{+}^N$ and proportional to
\be
\prod_{i<j}|\lambda_i-\lambda_j|^{\beta}\prod_{k}\lambda_k^{\beta N/2}e^{-\beta N\lambda_k/2}\ ,\label{eq:jpdf}
\ee
where $\beta\in\{1,2,4\}$ according to the physical symmetries of the cavity (hereafter the delay times $\tau_k$ are measured in units of the Heisenberg time $\tau_H$). The above joint distribution is valid in presence of completely transparent contacts between the cavity and the external world. The situation of perfect coupling is indeed the most relevant, because it is possible to reduce arbitrary non-ideal coupling to the ideal case using the procedure in \cite{Savin01,Savin01b}.

Our main result gives an asymptotic expansion as $N \to \infty$ for the average \emph{moments} ($k\geq0$)
\be
\tau_{k}^{(\beta)}= N^{k-1}\E[ \Tr Q^k]=\E\left[\frac{1}{N}\sum_{i=1}^N\lambda_i^{-k}\right]\ ,\label{eq:def1}
\ee
where the expectation $\mathbb{E}[\cdot]$ is taken with respect to \eqref{eq:jpdf}.
In particular, we derive explicit recurrence relations  which efficiently provide the coefficients of the expansion to all orders in $1/N$. 
\begin{theorem}[\cite{Simm12}] For all three symmetry classes $\beta\in\{1,2,4\}$ the following asymptotic expansion holds
\be
\tau_{k}^{(\beta)}=\sum_{g=0}^{\infty}\tau_{k,g}^{(\beta)}N^{-g}. \label{eq:sympexp}
\ee
\end{theorem}It is tempting to call~\eqref{eq:sympexp} a `genus' expansion. For complex Gaussian Hermitian matrices (GUE ensemble) the large-$N$ expansion of the moments indeed enumerates maps of given genus~\cite{Ercolani03,Zvonkin97}. Although we cannot prove that our expansion is related to enumeration of maps, we have strong evidence of an underlying enumeration problem also for the moments of the ensemble modelling the Wigner-Smith time-delay matrix.  More precisely, in this paper we extend a previous conjecture~\cite{Cunden16} on the integrality of the large-$N$ expansion coefficients for the cumulants of $\Tr Q^k$ beyond the leading order. 

Clearly, $\tau_{0}^{(\beta)}=1$. It is known that the average Wigner delay time is $\tau_1^{(\beta)}=1$ for every $N$ and $\beta$. There is no general result for $\tau_{k}^{(\beta)}$ for $k>1$ and generic $\beta$.  For $\beta=1$ (systems with time reversal symmetry) a finite-$N$ formula is available, but is too lengthy to be reported here, see~\cite[Eq. (27)]{Simm11}.
For $\beta=2$ the problem somehow simplifies. First, the expansion~\eqref{eq:sympexp} contains only powers of $N^{-2}$ (that is, $\tau_{k,g}^{(2)}=0$ if $g$ is odd); remarkably, two explicit finite-$N$ formulae for $\tau_{k}^{(2)}$ are available in the literature. We report them here.
\begin{proposition}[Eq. (19) of~\cite{Simm11} and Eq. (6) of~\cite{Novaes15I}] The moments of the proper delay times for $\beta=2$ are
\begin{subequations}
\barr
\displaystyle\tau_{k}^{(2)}&=& \displaystyle\frac{N^{k-1}}{k}\sum_{j=0}^{N-1}\binom{k+j-1}{k-1}\binom{k+j}{k-1}
\frac{\Gamma(2N-k-j)}{\Gamma(N-j)}\frac{\Gamma(N+1)}{\Gamma(2N)} \label{eq:exactMS}\\
&=& \displaystyle\frac{N^{k-1}}{k!}\sum_{j=0}^{k-1}(-1)^j\binom{k-1}{j}
\frac{\Gamma(N-j+k)}{\Gamma(N-j)}\frac{\Gamma(N+j+1-k)}{\Gamma(N+j+1)}\label{eq:exactNov}\ .
\earr
\end{subequations}
\end{proposition}
These two formulae have been derived independently with two different methods. Formula \eqref{eq:exactNov} is quite hard to extract large-$N$ asymptotics from (but this is possible, in principle, by using methods similar to those discussed in \cite{Vivo10,Krattenthaler10}). Formula \eqref{eq:exactMS} behaves better than \eqref{eq:exactNov}; although the number of terms in the sum is unbounded for large $N$, it is a sum of positive terms and the only asymptotic analysis one needs is the complete asymptotic series for the ratio of two Gamma functions. Along these lines of reasoning, the first three terms $\tau_{k,g}^{(2)}$ ($g=0,2,4$) of the large-$N$ expansion of $\tau_{k}^{(2)}$ have been obtained in \cite{Simm12}. We also recall that the leading order coefficients $\tau_{k,0}^{(\beta)}$ are independent of $\beta$ and given by the \emph{large Schr\"oder numbers}~\cite{Schroder}
\be
\tau^{(\beta)}_{k,0}={}_2F_1(1 - k,  k; 2; -1)\ ,
\ee 
where $_2F_1$ is the classical Gauss hypergeometric function, and whose generating function has been computed explicitly in~\cite{BerkKuip11,Cunden14}. A systematic study of the $1/N$-expansion for the time-delay matrix moments for $\beta=1$ and $\beta=2$ is our objective.

\subsection{Main results} Our central result is a recursion relation for the coefficients $\tau_{k,g}^{(\beta)}$ of the asymptotic expansion~\eqref{eq:sympexp}. As discussed in~\cite{Cunden16}, the Laguerre measure~\eqref{eq:jpdf} defines a \emph{one-cut $\beta$-ensemble}. In particular, as $N\to\infty$, the one-point marginal of~\eqref{eq:jpdf} concentrates on a single interval whose edges are $3\pm\sqrt{8}$. In this paper we denote the monic polynomial vanishing at the edges  by $y(z)=z^2-6z+1$. (This is the so called \emph{spectral curve} of the $\beta$-ensemble \eqref{eq:jpdf}.)

\begin{theorem}[Finite-$N$ recursion at $\beta=2$]\label{thm:recursionb2}
For every integer $k\geq1$, and every $N\geq1$,
\be
(N^{2}-k^{2})(k+1)\tau_{k+1}^{(2)}-3N^2(2k-1)\tau_{k}^{(2)}+N^2(k-2)\tau_{k-1}^{(2)}=0\ ,\label{eq:recformulab2}
\ee
with $\tau_{0}^{(2)}=\tau_{1}^{(2)}=1$.
\end{theorem}
The recursion formula~\eqref{eq:recformulab2} is much more efficient than~\eqref{eq:exactMS} and~\eqref{eq:exactNov} to generate tables of the moments of $Q$. See Appendix~\ref{app:num}. 
From Theorem~\ref{thm:recursionb2}
one obtains the recursion for the large-$N$ coefficients $\tau_{k,g}^{(2)}$ as a corollary.
\begin{cor}[Double-recursion for the large-$N$ coefficients at $\beta=2$]\label{cor:mainb2} The coefficients $\tau_{k,g}^{(2)}$ satisfy the homogeneous linear recurrence 
\be
(k+1)\tau^{(2)}_{k+1,g+2}-3(2k-1)\tau^{(2)}_{k,g+2}+(k-2)\tau^{(2)}_{k-1,g+2}-k^2(k+1)\tau^{(2)}_{k+1,g}=0\ ,\\
 \label{eq:recoeff}
\ee
for $g\geq0$ and $k\geq1$, with initial conditions 
\barr
\tau^{(2)}_{k,0}={}_2F_1(1 - k,  k; 2; -1),\quad \tau^{(2)}_{k,1}=0,\quad \tau^{(2)}_{0,g}=\delta_{0,g},\quad \tau^{(2)}_{1,g}=\delta_{0,g}\ .
\earr
In particular, all coefficients $\tau_{k,g}^{(2)}$ with odd  $g$ vanish identically.
\end{cor}

Analogous results can be obtained for $\beta=1$. In this case we found that $\tau^{(1)}_k$ (and therefore $\tau^{(1)}_{k,g}$) satisfy an inhomogeneous recursion.
\begin{theorem}[Finite-$N$ recursion  at $\beta=1$]\label{thm:recursionb1}
For every integer $k\geq1$, and every $N\geq1$,
\begin{subequations}
\be
(4k(k+1)+1-(N+1)^2)\tau_{k+1}^{(1)}+6N^2\tau_{k}^{(1)}-N^2\tau_{k-1}^{(1)}
=\frac{3}{k+1}((k+3N)Nb_k-N^2b_{k-1})\ ,\label{eq:recformulab1}
\ee
where the auxiliary sequence $b_{k}$ is the solution of 
\be
((N+1)^2-k^2)(k+1)b_{k+1}-(3N-1)(2k-1)Nb_{k}+(k-2)N^2b_{k-1}=0\ ,\label{eq:recformulab1bis}
\ee
\end{subequations}
with $\tau_{0}^{(1)}=\tau_{1}^{(1)}=1$, $b_0=\frac{N-1}{N}$ and $b_1=\frac{N-1}{N+1}$\ .
\end{theorem}
\begin{cor}[Double-recursion for the large-$N$ coefficients at $\beta=1$]\label{cor:mainb1} The coefficients $\tau_{k,g}^{(1)}$ satisfy the inhomogeneous linear recurrence 
\begin{subequations}
\be
\displaystyle\tau^{(1)}_{k+1,g+1}-6\tau^{(1)}_{k,g+1}+\tau^{(1)}_{k-1,g+1}+2\tau^{(1)}_{k+1,g}-4k(k+1)\tau^{(1)}_{k+1,g-1}\\
=\displaystyle\frac{3}{k+1}\left(b_{k-1,g+1}-3b_{k,g+1}-kb_{k,g}\right),\label{eq:recoeffb1}
\ee
where the auxiliary sequence $b_{k,g}$ is the solution of 
\be
b_{k+1,g+1}-3\frac{2k-1}{k+1}b_{k,g+1}+\frac{k-2}{k+1}b_{k-1,g+1}+2b_{k+1,g}+\frac{2k-1}{k+1}b_{k,g}+(1-k^2)b_{k+1,g-1}=0,\label{eq:recoeffb1bis}
\ee
\end{subequations}
for $g\geq0$ and $k\geq1$. The initial conditions are 
\barr
\tau^{(1)}_{k,0}&=&{}_2F_1(1 - k,  k; 2; -1),\quad \tau^{(1)}_{k,1}=\frac{1}{k!}\partial_z^{(k)}\left(\frac{1-3z-\sqrt{y(z)}}{2y(z)}\right)\Big|_{z=0}\ ,\\
 \tau^{(1)}_{0,g}&=&\delta_{0,g},\quad\tau^{(1)}_{1,g}=\delta_{0,g}\ ,\\
b_{k,0}&=&{}_2F_1(1 - k,  k; 2; -1),\quad
b_{k,1}=\frac{1}{k!}\partial_z^{(k)}\left(-\frac{z+1+3\sqrt{y(z)}}{2\sqrt{y(z)}}\right)\Big|_{z=0}\ ,\\
b_{0,g}&=&\delta_{0,g}-\delta_{1,g},\quad
b_{1,g}=(-1)^g(2-\delta_{0,g})\ .
\earr
\end{cor}
Again, the recursive formulae~\eqref{eq:recformulab1}-\eqref{eq:recformulab1bis} and~\eqref{eq:recoeffb1}-\eqref{eq:recoeffb1bis} can efficiently generate tables of moments and their large-$N$ expansion. See Appendix~\ref{app:num}. 

The proof of Theorem~\ref{thm:recursionb2} and Theorem~\ref{thm:recursionb1} is given in Section~\ref{sec:Wishart}. 
\subsection{Generating functions}
In this section, we derive explicit formulae for the generating functions of $\tau^{(\beta)}_{k,g}$.
Let us consider the formal power series 
\be
\varphi^{(\beta)}(z,\zeta)=\sum_{k,g= 0}^{\infty} \tau^{(\beta)}_{k,g}z^k\zeta^{g}.\label{eq:varphi}
\ee
Using the recursions of Corollary~\ref{cor:mainb2} and Corollary~\ref{cor:mainb1} it is possible to obtain a differential equation for $\varphi^{(\beta)}(z,\zeta)$.
For instance, for $\beta=2$, the generating function $\varphi^{(2)}(z,\zeta)$ satisfies
\be
\zeta^2 z^2 \varphi^{(2)}_{zzz}+\zeta^2 z \varphi^{(2)}_{zz}-y(z)\varphi^{(2)}_{z}+\frac{y'(z)}{2}\varphi^{(2)} +4= 0\ ,\label{eq:3rd}
\ee
but this third-order inhomogeneous differential equation~\eqref{eq:3rd} is not as tractable. (This was to be expected, since $\varphi^{(\beta)}(z,\zeta)$ is only a formal power series.) To make some further progress in the problem we introduce  the `partial' generating functions
\barr
F_{g}^{(\beta)}(z)&=&\sum_{k=0}^{\infty}\tau^{(\beta)}_{k,g}z^k\ ,\\
J_{k}^{(\beta)}(\zeta)&=&\sum_{g=0}^{\infty}\tau^{(\beta)}_{k,g}\zeta^g\ .
\earr
The series $\varphi^{(\beta)}(z,\zeta)$ , $F_{g}^{(\beta)}(z)$ and $J_{k}^{(\beta)}(\zeta)$ are of course related by
\barr
\varphi^{(\beta)}(z,\zeta)=\sum_{g=0}^{\infty}F_{g}^{(\beta)}(z)\zeta^g=\sum_{k=0}^{\infty}J_{k}^{(\beta)}(\zeta)z^k\ .
\earr
\begin{rmk}
Note that $\varphi^{(\beta)}(z,N^{-1})$ is the generating function of the finite-$N$ moments $\tau_{k}^{(\beta)}$ and $J_k^{(\beta)}(N^{-1})=\tau_{k}^{(\beta)}$. The partial generating functions $F_{g}^{(\beta)}(z)$ are central objects in the perturbative semiclassical approach. 
\end{rmk}
The partial generating functions have a remarkably simple (i.e. algebraic)  structure. We discuss first the case $\beta=2$.
\begin{cor}\label{cor:rec2simp1} If $g$ is odd, $F_{g}^{(2)}(z)=0$. For $g$ even, the generating function $F_{g}^{(2)}(z)$ satisfies
\be
\begin{cases}
F_{g+2}^{(2)}(z) = \sqrt{y(z)}\displaystyle\int_{0}^{z}\frac{\mathrm{d} x}{y(x)^{3/2}}\left\{x^{2}F_{g}^{(2)\prime\prime\prime}(x)+xF_{g}^{(2)\prime\prime}(x)\right\},\\
F_{0}^{(2)}(z) =\displaystyle\frac{3-z-\sqrt{y(z)}}{2}\ ,
\end{cases} \label{eq:diff1}
\ee
and has the following functional form ($g\geq1$)
\begin{equation}
F_{g}^{(2)}(z) = \displaystyle\frac{R_{g}(z)}{y(z)^{(3g-1)/2}},\label{eq:R_g}\\
\end{equation}
where $R_{g}(z)$ is a polynomial of degree $2g-2$. 
\end{cor}
\begin{cor}\label{cor:rec2simp}
The order $k\geq 2$ generating function $J_k^{(2)}(\zeta)$ is a rational function of the form
\be
J_k^{(2)}(\zeta)=\frac{P_k(\zeta^2)}{\prod_{j=0}^{k-1}\left(1-j^2\zeta^2\right)},\label{eq:ans}
\ee
($J_0(\zeta)=J_1(\zeta)=1$), where $P_k(\zeta)$ is a polynomial satisfying the three term recursion
\be
\begin{cases}
kP_k(\zeta)-3(2k-3)P_{k-1}(\zeta)+(k-3)(1-(k-2)^2\zeta)P_{k-2}(\zeta)=0\ ,\\
P_0(\zeta)=P_1(\zeta)=1\ .
\end{cases} \label{eq:reJk}
\ee  
\end{cor}
\begin{proof}[Outline of the proof] The proof of Corollary~\ref{cor:rec2simp} goes as follows. First, one multiplies the recursion~\eqref{eq:recoeff} by $\zeta^{g+1}$ and sums over $g$ to obtain a recursion for $J_k^{(2)}(\zeta)$. Inserting the expression~\eqref{eq:ans} in the obtained recursion,  elementary steps provide~\eqref{eq:reJk}. It remains to be proved that $P_k(\zeta)$ is a polynomial. In fact $P_0(\zeta)=P_1(\zeta)=1$ are polynomials. From~\eqref{eq:reJk} it also follows  that if $P_{k-1}(\zeta)$ and $P_{k-2}(\zeta)$ are polynomials so is $P_k(\zeta)$.   This completes the proof.
The proof of Corollary~\ref{cor:rec2simp1} is again a routine calculation (multiplication of~\eqref{eq:recoeff} by $z^k$ and sum over $k$). However, in this case $R_g(z)$ satisfies a recursion relation too complicated to be reported here. Nevertheless, it is easy to see that $R_{g+2}(z)$ is a polynomial if $R_g(z)$ is so, and to compute its degree. The details are omitted.
 \end{proof}

From the partial generating functions $F^{(2)}_g$ and $J^{(2)}_k$ it is possible to extract estimates on $\tau_{k,g}^{(2)}$ as $k\to\infty$ (resp. $g\to\infty$) with $g$ (resp. $k$) fixed. These asymptotics results are based on Darboux's method~\cite{Darboux} (see the statement in~\cite[Theorem 11.3]{Odlyzko}).
\begin{cor}\label{cor:asymp} The following asymptotics hold
\be
\tau_{k,2g}^{(2)}\sim
\begin{cases}
A_g\, k^{\frac{6g-3}{2}}(3-\sqrt{8})^{-k}&\text{as $k\to\infty$ with $g\geq1$ fixed},\\
B_k\, (k-1)^{2g}&\text{as $g\to\infty$ with $k\geq1$ fixed}.
\end{cases}
\ee 
The constants $A_g$ and $B_k$ are given explicitly by
\barr
A_g&=&\displaystyle\frac{\left(\sqrt{32}\left(3-\sqrt{8}\right)\right)^{\frac{1-6g}{2}}R_{2g}\left(3-\sqrt{8}\right)}{\Gamma\left(\frac{6g-1}{2}\right)}\ ,\\
B_k&=&\frac{P_k\left((k-1)^{-2}\right)}{\prod_{j=0}^{k-2}\left(1-j^2(k-1)^{-2}\right)}\ .
\earr
\end{cor}

 It is possible to obtain similar recursions for  generating functions in the orthogonal case $\beta=1$. Below we write the relation satisfied by $F_{g}^{(1)}(z)$ explicitly. In this case,  the recurrence relation for $\tau_{k,g}^{(1)}$ is not homogeneous and involves the auxiliary sequence $b_{k,g}$. Therefore, the recursion for $F_{g}^{(1)}(z)$ is coupled to a (homogeneous) recursion for the auxiliary generating function $f_{g}(z)=\sum_{k}b_{k,g}z^k$. This fact complicates the structure of the formulae; it turns out that the generating functions $F^{(1)}_g(z)$ are algebraic functions but they do not have the simple functional form \eqref{eq:R_g} of $F^{(2)}_g(z)$. From a purely algorithmic point of view this is not a problem.
\begin{cor}\label{cor:rec2simp1b1} The generating functions $F_{g}^{(1)}(z)$ satisfy for $g\geq1$
\begin{subequations}
\be
\begin{cases}
F^{(1)}_{g+1}(z)=\displaystyle\frac{1}{y(z)}\displaystyle\int_{0}^z\!\!\mathrm{d}x\left\{4x^2F^{(1)\prime\prime\prime}_{g-1}(x)+8xF^{(1)\prime\prime}_{g-1}(x)-2F^{(1)\prime}_g(x)+3(x-3)f_{g+1}(x)-3xf'_g(x)\right\}\ ,\\
\\

F_{0}^{(1)}(z)=\displaystyle\frac{3-z-\sqrt{y(z)}}{2},\quad F_{1}^{(1)}(z) =\displaystyle\frac{1-3z-\sqrt{y(z)}}{2y(z)}\ ,
\end{cases} \label{eq:diff1b1}
\ee
where the functions $f_g(z)$ satisfy 
\be
\begin{cases}
f_{g+1}(z)=\sqrt{y(z)}\displaystyle\int_{0}^z\frac{\mathrm{d}x}{y(x)^{3/2}}\left\{f_g(x)-2(x+1)f'_g(x)+x^2f'''_{g-1}(x)+xf''_{g-1}(x)-f'_{g-1}(x)\right\}\ ,\\\\
f_0(z)=\displaystyle\frac{3-z-\sqrt{y(z)}}{2},\quad f_1(z) =\displaystyle-\frac{z+1+\sqrt{y(z)}}{2\sqrt{y(z)}}\ .
\end{cases} \label{eq:diff1b1}
\ee
\end{subequations}
\end{cor}
The generating functions can be computed systematically and efficiently on standard computer algebra packages. 
The first few functions $F^{(\beta)}_g(z)$ for $\beta=1$ and $\beta=2$ are reported in Appendix~\ref{app:num}.  We note that only the leading order (planar) $g=0$ and first four corrections $g=1,\dots,4$ have appeared in the literature so far. (See~\cite{Simm12} for the random matrix approach and \cite{BerkKuip10,BerkKuip11,Sieber14} for semiclassical techniques.) 
\section{Integrality conjecture and its heuristic explanation}
The inspection of the first values of $\tau_{k,g}^{(\beta)}$ for $\beta=1$ and $\beta=2$ (see Table~\ref{tab:I}) suggests that they are positive integers. A similar fact has been recently observed~\cite{Cunden14,Cunden16,Simm13} for the leading order in $1/N$ of higher order cumulants of $\Tr Q^{k}$ (covariance, third order cumulants, etc.). 
\begin{conj}\label{conj:1} For $\beta=1$ and $2$, $\tau_{k,g}^{(\beta)}\in\N$ for every $k$ and $g$.
\end{conj}
This conjecture extends beyond the leading order a generic conjectural statement~\cite{Cunden16} for the cumulants of $\Tr Q^k$ at generic $\beta$.
We have considerable evidence supporting the conjecture. In fact, for $\beta=2$, using the functional form of the generating functions we can actually prove that \emph{infinitely many} $\tau_{k,g}^{(2)}$'s are positive integers.
We proceed to prove the following.
\begin{theorem}\label{thm:conj} $\tau_{k,g}^{(2)}\in\N$ for all $k\leq k^{\star}$ and for all $g\leq g^{\star}$ where $k^{\star}=10000$ and $g^{\star}=80$.
\end{theorem} 
\begin{proof}  The proof is based on the partial generating functions $J^{(2)}_k(\zeta)$ and $F^{(2)}_g(z)$. For a fixed $k$, in order to prove that $\tau^{(2)}_{k,g}\in\N$ for all $g$, we consider the generating function  $J_k(\zeta)$. The representation~\eqref{eq:ans}
\be
J^{(2)}_k(\zeta)=\frac{P_k(\zeta^2)}{\prod_{j=0}^{k-1}\left(1-j^2\zeta^2\right)}
\ee
shows that a sufficient condition for $\tau_{k,g}^{(2)}$ to be nonnegative integers is that the polynomial $P_k(\zeta)$ has nonnegative integer coefficients (the series expansion of $\prod_{j}\left(1-j^2\zeta^2\right)^{-1}$ at $\zeta=0$ is a product of geometric series), as suggested by the inspection of the first few polynomials (see Appendix~\ref{app:num}). Therefore, the claim `$\tau^{(2)}_{k,g}\in\N$ for all $g$' involving an \emph{infinite} number of coefficients can be proved by exhaustion of a \emph{finite} number of cases: first, one computes the polynomial $P_k(\zeta)$ using the recursion~\eqref{eq:reJk}; then, one checks that  the finitely many coefficients of  $P_k(\zeta)$ are all nonnegative integers. This can be easily done by using a symbolic algebra software. (We have run a Maple code to compute recursively $P_k(\zeta)$ and verify that it has nonnegative integer coefficients for all $k\leq k^{\star}$.)

The proof that if $g\leq g^{\star}$ then $\tau_{k,g}^{(2)}\in\N$ for all $k$ goes along similar lines. For $g$ odd the proof is trivial. Let us consider, for $g$ even, the partial generating functions 
\be
F^{(2)}_{g}(z) = \displaystyle\frac{R_{g}(z)}{y(z)^{(3g-1)/2}}\ .\\
\ee
Here we use the classical identity
\be
\frac{1}{\sqrt{z^2-2tz+1}}=\sum_{\ell=0}^{\infty}p_{\ell}(t)z^{\ell}\ ,\label{eq:Leg}
\ee
where $p_{\ell}(t)$ is the Legendre polynomial of degree $\ell$. For $t=3$:
\be
p_{\ell}(3)=\sum_{p=0}^{{\ell}}\binom{{\ell}}{p}^2 2^p\in\N\ .\label{eq:Leg2}
\ee
Hence we can focus on the coefficients of the polynomial $R_g(z)$. This case, however, is complicated by the fact that, although these coefficients seem to be integers, they are not necessarily positive (see Appendix~\ref{app:num}). Nevertheless, we can take advantage of~\eqref{eq:Leg}-\eqref{eq:Leg2} as follows. We have
\be
\frac{1}{\left(z^2-6z+1\right)^{(3g-1)/2}}=\sum_{\ell=0}^{\infty}C_{\ell}z^{\ell},\quad\text{with}\quad C_{\ell}=\sum_{\substack{\ell_1,\dots,\ell_{3g-1}=0\\ \ell_1+\cdots+\ell_{3g-1}=\ell}}^{\infty}p_{\ell_1}(3)\cdots p_{\ell_{3g-1}}(3)\in\N\ .
\ee
Note that $C_{\ell+1}\geq C_{\ell}$. If we denote  $R_g(z)=\sum_{j=0}^{2g-2}a_{g,j}z^j$, then
\be
F_{g}(z) = \sum_{k=0}^{\infty}\left(\sum_{j=0}^{2g-2}a_{g,j}C_{k-j}\right)z^k\ .
\ee
Since $C_{\ell}$ is nondecreasing we have
\be
\tau_{k,g}^{(2)}=\sum_{j=0}^{2g-2}a_{g,j}C_{k-j}\geq C_{k-(2g-2)}\sum_{j=0}^{2g-2}a_{g,j}\ .
\ee
Therefore, we conclude that the two conditions i) $a_{g,j}\in\Z$ and ii) $\sum_{j}a_{g,j}\geq0$ imply $\tau_{k,g}^{(2)}\in\N$. Again, these conditions can be verified case by case using symbolic algebra softwares. 
\end{proof}
Obviously, the value $k^{\star}$ and $g^{\star}$ in Theorem~\ref{thm:conj} are fixed by limited computational power. 

\subsection{Semiclassical explanation of the conjecture}
Periodic orbit theory is a collection of diversified results in the semiclassical analysis of quantum systems. Since its creation \cite{Gutzwiller71,Hannay84,Richter01}, the theory has played an important role in the mathematical investigations of quantum chaos. 
Later, the ideas of periodic orbit theory were adapted to study quantum transport in the chaotic regime. Not surprisingly, a scattering orbit approach to delay times has been also developed~\cite{BerkKuip10,BerkKuip11,Sieber14,Novaes15II}. The semiclassical approach  is formulated in terms of the classical trajectories connecting the exterior and interior regions of the cavity. Each observable (e.g. $\tau_k^{(\beta)}$) is written as a sum over classical trajectories of wave amplitudes. The semiclassical contribution of a trajectory is determined by its topological properties. Therefore, the set of trajectories is partitioned according to topological properties where each class of trajectories is represented by a \emph{diagram} with a given number of \emph{incoming channels}, \emph{links} and \emph{encounters}. See the recent paper by Kuipers et al.~\cite{Sieber14} for details. It turns out that the semiclassical contribution of a class of trajectories represented by a diagram $\mathcal{D}$ is given by $(-1)^{c_1(\mathcal{D})}N^{c_2(\mathcal{D})}$, where $c_{1,2}(\mathcal{D})\in\Z$. Then, one should sum over all classes of admissible trajectories (a sum over diagrams). It turns out that the admissible diagrams depend on the presence or absence of time-reversal symmetry (the Dyson index $\beta$ in random matrix theory). The sum over diagrams is usually the hard part of the semiclassical approach; however, since the contribution of each diagram is given by the number of scattering channels $N$ to some integer power, it is clear that the coefficients in the $1/N$-expansion of $\tau_k^{(\beta)}$ are integers. We also mention that Novaes~\cite{Novaes11} computed the leading order of $\tau_{k}^{(\beta)}$ by considering the asymptotics of Selberg-like integrals; his method consists in enumerating certain classes of lattice paths. As expected, those paths were found to be in bijection with Schr\"oder paths.

\section{Inverse moments of Wishart-Laguerre matrices}
\label{sec:Wishart} 
In this section we present several new results on the moments of inverse Wishart matrices. According to~\eqref{eq:jpdf}-\eqref{eq:def1}, the moments of the Wigner-Smith time-delay matrix $Q$ are related to the inverse moments of a set of random variables belonging to a specific Laguerre ensemble (see Remark~\ref{rmk:our} below).

\subsection{Technical results} Let $W_N$ be a $N\times N$ random matrix distributed according to the Wishart-Laguerre density

\begin{equation}
P_{\beta,\alpha}(W_N) \propto e^{-\frac{\beta}{2}\tr(W_N)}\det (W_N)^{\frac{\beta}{2}(\alpha+1)-1}\ , \label{eq:wishart}
\end{equation}
where $\alpha$ is a generic parameter satisfying $\Re(\alpha) > 0$. 
This density is defined on the space of $N\times N$ positive definite real symmetric and complex Hermitian matrices for $\beta=1$ and $\beta=2$, respectively. We are interested in computing the moments ($k\in\Z$) and their large-$N$ expansions
\barr
D^{(\beta)}_N(k,\alpha) &=& 
\mathbb{E}[\mathrm{Tr}W_N^{k}]\ ,\label{eq:largeND} 
\earr
where from now on $\mathbb{E}[\cdot]$ denotes averaging with respect to \eqref{eq:wishart}, and $\alpha$ is large enough to ensure that~\eqref{eq:largeND} is finite.

\begin{rmk}\label{rmk:our} For our physical application on the Wigner-Smith time delay matrix we have $\tau_{k}^{(\beta)}=N^{k-1}D^{(\beta)}_N(-k,N+2-\beta)$, for $\beta=1$ and $2$. 
\end{rmk}
Let us define the generating function
\begin{equation}
M_N^{(\beta)}(s) = \mathbb{E}\left[\mathrm{Tr}(W_Ne^{sW_N})\right]\ , \label{GF}
\end{equation}
where in what follows $s \leq 0$. For simplicity, the dependence of~\eqref{GF} on $\alpha$ is omitted. First we present a lemma.
\begin{lemma}
	\label{lem:1}
The generating function $M_N^{(\beta)}(s) $ is related to moments via the identity
\begin{equation}
D^{(\beta)}_N(k,\alpha)=
\begin{cases}
\displaystyle\frac{\partial^{k-1}M_N^{(\beta)}(s) }{\partial s^{k-1}}\biggl|_{s\to 0^{-}}&\text{if $k>0$}\ ,\\
\displaystyle\frac{(-1)^{|k|}}{|k|!}\int_{-\infty}^{0}M_N^{(\beta)}(s) s^{|k|}\,\mathrm{d} s &\text{if $k\leq0$}\ . \label{gfident}
\end{cases}
\end{equation} 
\end{lemma}
The next ingredient is that $M_N^{(\beta)}(s) $ satisfies a differential equation.
\begin{theorem}[Generating function for $\beta=2$, adaptation of Theorem 6.4 in \cite{Haagerup03}]
\label{thm:Haagerup}
For all $N\geq1$, $M_N^{(2)}$ satisfies the following homogeneous second-order differential equation
\begin{equation}
s(1-s^{2})M_N^{(2)\prime\prime}+(3-2(\alpha+2N)s-5s^{2})M_N^{(2)\prime}-(3(\alpha+2N)+4s-\alpha^{2}s)M_N^{(2)} = 0\ . \label{ode}
\end{equation}
\end{theorem}
\begin{theorem}[Generating function for $\beta=1$]
\label{thm:orthogonal}
For all $N\geq1$, $M_N^{(1)}$ satisfies the following inhomogeneous second-order differential equation
\begin{equation}
(4s^{3}-s)M_N^{(1)\prime\prime}+s(16s+2(\alpha-1)+4N)M_N^{(1)\prime}+s(9-\alpha^{2})M_N^{(1)}=
(3s+3)M_{N-1}^{(2)\prime}-(3\alpha+6N-6)M_{N-1}^{(2)}\ . \label{odeb1}
\end{equation}
\end{theorem}
From these differential equations, it is easy to get a recurrence relation on the moments $D^{(\beta)}_N(k,\alpha)$ using Lemma~\ref{lem:1}. Indeed, for $k<0$, multiplying \eqref{ode}-\eqref{odeb1} by $s^{k}$ and integrating from $-\infty$ to $0$, one removes all mention of $M^{(\beta)\prime\prime}_N(s)$ and $M^{(\beta)\prime}_N(s)$ by successive integrations by parts and then applies identity~\eqref{gfident} to compute all integrals. Similarly, for $k>0$, one differentiates  \eqref{ode}-\eqref{odeb1} $k$ times and then applies~\eqref{gfident}. Several cancellations simplify the outcome considerably. The final results are the following recurrence relations. 
\begin{theorem}[Moments of complex Wishart matrices, $\beta=2$]
The moments $D^{(2)}_N(k,\alpha)$ satisfy
\begin{equation}
(k+2)D^{(2)}_N(k+1,\alpha)-(2k+1)(\alpha+2N)D^{(2)}_N(k,\alpha)-(k-1)(k^{2}-\alpha^{2})D^{(2)}_N(k-1,\alpha)=0\ . \label{eq:recusionalpha}
\end{equation}
\end{theorem}
\begin{theorem}[Moments of real Wishart matrices, $\beta=1$]
The moments $D^{(1)}_N(k,\alpha)$ satisfy
\barr
D^{(1)}_N(k+1,\alpha)-(2(\alpha-1)+4N)D^{(1)}_N(k,\alpha)-(1-\alpha^2+4k(k-1))D^{(1)}_N(k-1,\alpha)\nonumber\\
=\frac{3}{k-1}((\alpha+2N-k-1)D^{(2)}_{N-1}(k,\alpha)-D^{(2)}_{N-1}(k+1,\alpha))\ .
 \label{eq:recusionalphab1}
\earr
\end{theorem}
\begin{rmk}
Haagerup and  Thorbj{\o}rnsen~\cite[Theorem 8.5]{Haagerup03} proved~\eqref{ode} and then deduced the finite-$N$ recurrence~\eqref{eq:recusionalpha} for positive moments $k>0$,  thus generalizing the Harer-Zagier
recursion formula for GUE matrices~\cite{Harer86} to the complex Wishart ensemble. Here we show that the same recursion holds for negative moments ($k\leq0$). The differential equation~\eqref{odeb1} and the recursion~\eqref{eq:recusionalphab1} for $\beta=1$ are new results. The inhomogeneous term in the differential equation~\eqref{odeb1} for $\beta=1$ is related to a $\beta=2$ ensemble; the reason for this will become clear in the proof below (see Eq.~\eqref{eq:Masrho}-\eqref{beta1dens}). Theorem~\ref{thm:recursionb2} is a specialization of~\eqref{eq:recusionalpha} renaming $\tau_{k}^{(2)}= N^{k-1}D_N^{(2)}(-k,\alpha=N)$. Theorem~\ref{thm:recursionb1} is a specialization of~\eqref{eq:recusionalphab1} renaming $\tau_{k}^{(1)}= N^{k-1}D^{(1)}_N(-k,\alpha=N+1)$ and $b_k=N^{k-1}D^{(2)}_{N-1}(-k,\alpha=N+1)$.

In both real and complex cases, traces of powers of Wishart matrices were considered for generic covariance matrix $\Sigma$ using orthogonal and unitary Weingarten functions \cite{Letac,Matsumoto12}. However, even in the simplest case $\Sigma=I$ (the one considered here), the Weingarten functions are not easy to compute. Furthermore, it is not clear how to obtain an asymptotic $1/N$-expansion from such formulae.
\end{rmk}
\begin{proof}[Proof of Lemma~\ref{lem:1}] For $k>0$, Eq.~\eqref{gfident} is a classical formula. Let us consider the case of negative moments. By a unitary transformation we can write $W_N = U\Lambda U^{\dagger}$ where $\Lambda$ is diagonal, so that $W_Ne^{sW_N}$ = $U\Lambda e^{s\Lambda }U^{\dagger}$ and the left-hand side of \eqref{gfident} is
	\begin{equation}
	\mathrm{Tr}\left(U\int_{-\infty}^{0}\Lambda e^{s\Lambda}s^{k}\,\mathrm{d}s\,U^{\dagger}\right) = \mathrm{Tr}\left(\int_{-\infty}^{0}\Lambda e^{s\Lambda }s^{k}\,\mathrm{d}s\,\right)\ ,
	\end{equation}
	where the integral acts on each diagonal entry via
	\begin{equation}
	\int_{-\infty}^{0}\lambda_{j}e^{s\lambda_{j}}s^{k}\,\mathrm{d}s = \frac{1}{\lambda_{j}^{k}}(-1)^{k}k!\ .
	\end{equation}
	Then the trace is the sum over $j$ and one obtains the right-hand side of \eqref{gfident}.
\end{proof}

\subsection{Proof of the technical results}
\label{sub:proof}
\begin{proof}[Proof of Theorems~\ref{thm:Haagerup} and~\ref{thm:orthogonal}] 
We introduce the finite-$N$ average density of eigenvalues of $W_{N}$
\be
\rho^{(\beta)}_{N}(x) =\E\left[\frac{1}{N}\sum_{i=1}^N\delta(x-\lambda_i)\right]\ ,
\ee
where the expectation is taken with respect to $P_{\beta,\alpha}$ in \eqref{eq:wishart}. The generating function $M_N^{(\beta)}(s)$ can be written as
\be
M_N^{(\beta)}(s) = \int_{0}^{\infty}xe^{sx}\rho_{N}^{(\beta)}(x)\,\mathrm{d}x\ . \label{eq:Masrho}
\ee
The fundamental ingredient is that for $\beta=2$ and $\beta=1$ the finite-$N$ eigenvalue density $\rho^{(\beta)}_{N}(x)$ is given explicitly in terms of Laguerre polynomials. We denote the standard Laguerre polynomial of degree $N$ and parameter $\alpha$ by
\begin{equation}
L^{(\alpha)}_{N}(x) = \sum_{i=0}^{N}(-1)^{i}\binom{N+\alpha}{N-i}\frac{x^{i}}{i!}\ .
\end{equation}
They satisfy the second order differential equation
\begin{equation}
xL_N^{(\alpha)\hspace{1pt}\prime\prime}(x)+(1-x+\alpha)L_N^{(\alpha)\hspace{1pt}\prime}(x)+NL^{(\alpha)}_{N}(x)=0\ . \label{laguerreode}
\end{equation}
The following formulae for the mean eigenvalue densities are well-known. For $\beta=2$, the Christoffel-Darboux formula leads to the expression
\begin{equation}
\rho^{(2)}_{N}(x) = \frac{\Gamma(N+1)}{\Gamma(N+\alpha)}(L^{(\alpha)}_{N}(x)L^{(\alpha)\hspace{1pt} \prime}_{N-1}(x)-L^{(\alpha)\hspace{1pt} \prime}_{N}(x)L^{(\alpha)}_{N-1}(x))x^{\alpha}e^{-x}\ . \label{beta2dens}
\end{equation}
For $\beta=1$, explicit forms for the mean eigenvalue density for the Laguerre ensembles were given in \cite{FNH99,Wid99}. The result is that $\rho^{(1)}_{N}(x)$ can be expressed in terms of $\rho^{(2)}_{N-1}(x)$ plus a correction term. The correction term was understood in a more general context in \cite{AFNM00}, where it was formulated in a slightly different way that turns out to be useful here. Moreover, the correction term depends on the parity of $N$. Nevertheless, the moments $\tau_{k}^{(1)}$ are always rational functions of $N$, and hence they are uniquely determined by subsequences like even $N$'s. For simplicity, in the case $\beta=1$ we will perform our computations for $N$ even, but our final results do not depend on the parity of $N$. From \cite{AFNM00}, for $N$ even we have
\begin{equation}
\rho^{(1)}_{N}(x) = \rho^{(2)}_{N-1}(x)-d_{N}x^{(\alpha-1)/2}e^{-x/2}L^{(\alpha)}_{N-2}(x)\psi(x)\ , \label{beta1dens}
\end{equation}
where the constant $d_{N}$ is given by
\begin{equation}
d_{N} = \frac{1}{4}\frac{\Gamma(N)}{\Gamma(\alpha+N-1)}
\end{equation}
and
\begin{equation}
\psi(x) = \int_{0}^{\infty}\mathrm{sgn}(x-y)y^{(\alpha-1)/2}e^{-y/2}L^{(a)}_{N-2}(y)\,\mathrm{d}y\ .
\end{equation}
The plan is now to insert these expressions into \eqref{eq:Masrho} and derive differential equations for the generating functions.
\flushleft \textbf{$\beta=2$ : Derivation of equation \eqref{ode}}\\
\flushleft Equation \eqref{ode} has been derived in at least two ways in the literature, first by Haagerup and Thorbj{\o}rnsen \cite{Haagerup03} and then rediscovered in a more general setting by Ledoux \cite{Led04}. We will sketch below the proof given in \cite{Haagerup03}, where the idea is to prove that
\begin{equation}
M^{(2)}_N(s) =N(\alpha+N)_2F_1(1-\alpha-N,1-N;2;s^{2})(1-s)^{-(\alpha+2N)}\ . \label{mgfgauss}
\end{equation}
It satisfies a classical second-order differential equation, which after some lengthy algebraic manipulations yields \eqref{ode}. To prove \eqref{mgfgauss}, a crucial role is played by the classical second-order differential equation \eqref{laguerreode} satisfied by the Laguerre polynomials.
Combining the differential equation with the expression \eqref{beta2dens} leads to
\begin{equation}
\frac{\mathrm{d}}{\mathrm{d}x}\left(x\rho^{(2)}_{N}(x)\right) = \sqrt{N(N+\alpha)}L^{(\alpha)}_{N}(x)L^{(\alpha)}_{N-1}(x)x^{\alpha}e^{-x}\ .
\end{equation}
Then integration by parts in \eqref{eq:Masrho} results in
\begin{equation}
M^{(2)}_{N}(s) = -\frac{\sqrt{N(N+\alpha)}}{s}\int_{0}^{\infty}e^{-(1-s)x}x^{\alpha}L^{(\alpha)}_{N}(x)L^{(\alpha)}_{N-1}(x)\,\mathrm{d}x\ . \label{intbeta2}
\end{equation}
To compute the integral in \eqref{intbeta2}, substitute $u=x(1-s)$ and make use of the scaling identity
\begin{equation}
L^{(\alpha)}_{N}(cx) = \sum_{r=0}^{N}\binom{N+\alpha}{N-r}c^{r}(1-c)^{N-r}L^{(\alpha)}_{r}(x)\ ,
\end{equation}
with $c = (1-s)^{-1}$, which reduces the integral to the orthogonality relation of Laguerre polynomials. The single summation that remains is recognised as the series definition of the hypergeometric function, and hence \eqref{mgfgauss}.

\begin{flushleft} \textbf{$\beta=1$ : Derivation of equation \eqref{odeb1}}\\

The proof of Theorem~\ref{thm:orthogonal} for $\beta=1$ is given below and is based on the paper~\cite{Ledoux} where the analogous computation was done for the Gaussian Orthogonal Ensemble. It was suggested in \cite{Ledoux} that the computation done there could in principle be carried out for other classical ensembles, but to our knowledge this has not been done before. 
\end{flushleft}
The starting point is the finite-$N$ formula \eqref{beta1dens} which we insert into \eqref{eq:Masrho}:
\be
M_N^{(1)}(s)=M_{N-1}^{(2)}(s)+d_N Y(s)\ ,\label{eq:M1M2}
\ee
where
\begin{equation}
Y(s) = -\int_{0}^{\infty}e^{sx}x^{(\alpha+1)/2}e^{-x/2}L^{(\alpha)}_{N-1}(x)\psi(x)\,\mathrm{d}x\ .
\end{equation}
From Theorem \ref{thm:Haagerup}, we may treat $M_{N-1}^{(2)}(s)$ in \eqref{eq:M1M2} as a known quantity. Hence we seek a differential equation for the second addendum in \eqref{eq:M1M2}.  We write the differential equation \eqref{laguerreode} as the eigenfunction relation
\begin{equation}
-TL^{(\alpha)}_{N}(x) = NL^{(\alpha)}_{N}(x)\ , \label{eigenrelation}
\end{equation}
where $Tf = xf''+(1+\alpha-x)f'$. Next, for any sufficiently smooth $f$ and $g$, the following identity is a direct consequence of integration by parts:
\begin{equation}
\int_{0}^{\infty}f(-Tg)\,x^{\alpha}e^{-x}\,\mathrm{d}x = \int_{0}^{\infty}xf'g'\,x^{\alpha}e^{-x}\,\mathrm{d}x\ . \label{parts}
\end{equation}
We now derive a differential equation for $Y(s)$. Taking a derivative with respect to $s$ and integrating by parts yields
\begin{align}
\nonumber \frac{\mathrm{d} Y}{\mathrm{d}s} &= -\int_{0}^{\infty}e^{sx}x^{(\alpha+3)/2}e^{-x/2}L^{(\alpha)}_{N-1}(x)\psi(x)\,\mathrm{d}x\\
\nonumber &= \frac{(\alpha+3)/2}{s-1/2}\int_{0}^{\infty}e^{sx}x^{(\alpha+1)/2}e^{-x/2}L^{(\alpha)}_{N-1}(x)\psi(x)\,\mathrm{d}x\\
\nonumber &+\frac{1}{s-1/2}\int_{0}^{\infty}e^{sx}x^{(\alpha+3)/2}e^{-x/2}L^{(\alpha)\hspace{1pt}\prime}_{N-1}(x)\psi(x)\,\mathrm{d}x\\
&+\frac{2}{s-1/2}\int_{0}^{\infty}xe^{sx}x^{\alpha}e^{-x}L^{(\alpha)}_{N-1}(x)L^{(\alpha)}_{N-2}(x)\,\mathrm{d}x\ . 
\end{align}
In the last integral we used that
$\psi'(x) = 2x^{(\alpha-1)/2}e^{-x/2}L^{(\alpha)}_{N-2}(x)$. We can rewrite this as
\begin{equation}
(s-1/2)Y'(s) = -\frac{\alpha+3}{2}Y(s)+2u_{N-1}'(s)+\chi_{N}(s)\ , \label{1der}
\end{equation}
where
\begin{align}
u_{N-1}(s) &= \int_{0}^{\infty}e^{sx}L^{(\alpha)}_{N-1}(x)L^{(\alpha)}_{N-2}(x)x^{\alpha}e^{-x}\,\mathrm{d}x\ ,\\
\chi_{N}(s) &= \int_{0}^{\infty}e^{sx}x^{(\alpha+3)/2}e^{-x/2}L^{(\alpha)\hspace{1pt}\prime}_{N-1}(x)\psi(x)\,\mathrm{d}x\ .
\end{align}
Note that $u_{N-1}(s)$ is closely related to the Laplace transform of $\rho^{(2)}_{N-1}(x)$ (cf. identity \eqref{intbeta2})
\be
u_{N-1}(s) = \frac{-sM^{(2)}_{N-1}(s)}{4d_{N}}\ . \label{eq:u_n1}
\ee
Differentiating \eqref{1der} one more time we arrive at
\begin{equation}
(s-1/2)Y''(s)+Y'(s)=-\frac{\alpha+3}{2}Y'(s)+2u_{N-1}''(s)+\chi_{N}'(s)\ . \label{2der}
\end{equation}
On the other hand, we can use \eqref{eigenrelation} and \eqref{parts} to show that
\begin{align}
\nonumber -(N-1)Y(s) &= \int_{0}^{\infty}e^{sx}x^{(1-\alpha)/2}e^{x/2}(-TL^{(\alpha)}_{N-1}(x))\,\psi(x)\,x^{\alpha}e^{-x}\,\mathrm{d}x \\
\nonumber &=\int_{0}^{\infty}(e^{sx}x^{(1-\alpha)/2}e^{x/2}\psi(x))'L^{(\alpha)\hspace{1pt}\prime}_{N-1}(x)\,x^{\alpha+1}e^{-x}\,\mathrm{d}x\\
\nonumber &=(s+1/2)\int_{0}^{\infty}e^{sx}x^{(\alpha+3)/2}e^{-x/2}L^{(\alpha)\hspace{1pt}\prime}_{N-1}(x)\,\psi(x)\,\mathrm{d}x \\
 \nonumber &+\frac{1-\alpha}{2}\int_{0}^{\infty}e^{sx}x^{(\alpha+1)/2}e^{-x/2}L^{(\alpha)\hspace{1pt}\prime}_{N-1}(x)\,\psi(x)\,\mathrm{d}x \\
&+2\int_{0}^{\infty}xe^{sx}x^{\alpha}e^{-x}L^{(\alpha)\hspace{1pt}\prime}_{N-1}(x)L^{(\alpha)}_{N-2}(x)\,\mathrm{d}x\ . \label{otoh3}
\end{align}
Differentiating \eqref{otoh3} gives the identity
\begin{equation}
-(N-1)Y'(s) = \chi_{N}(s)+(s+1/2)\chi_{N}'(s)+\frac{-\alpha+1}{2}\chi_{N}(s)+2K_{N}(s)\ , \label{tau1prime}
\end{equation}
where
\begin{equation}
K_{N}(s) = \int_{0}^{\infty}x^{2}e^{sx}x^{\alpha}e^{-x}L'_{N-1}(x)L_{N-2}(x)\,\mathrm{d}x\ . \label{KN1}
\end{equation}
Solving \eqref{1der} and \eqref{2der} for $\chi_{N}(s)$ and $\chi_{N}'(s)$ and inserting the result into \eqref{tau1prime} gives
\begin{equation}
(s^{2}-1/4)Y''+\left(4s+\frac{\alpha-1}{2}+N\right)Y'+\frac{9-\alpha^{2}}{4}Y+(\alpha-3)u_{N-1}'-(2s+1)u_{N-1}''+2K_{N}=0\ . \label{kneqn}
\end{equation}
Our aim is now to express $K_{N}$ in terms of known quantities. Integrating by parts in \eqref{KN1}, we find
\begin{equation}
K_{N}(s) = -(2+\alpha)u_{N-1}'-(s-1)u_{N-1}''-\int_{0}^{\infty}x^{2}e^{sx}x^{\alpha}e^{-x}L_{N-1}(x)L'_{N-2}(x)\,\mathrm{d}x\ . \label{KN2}
\end{equation}
Adding the two representations \eqref{KN1} and \eqref{KN2} shows that
\begin{align}
&(\alpha-3)u_{N-1}'(s)-(2s+1)u_{N-1}''(s)+2K_{N}(s)\\
&=(\alpha-3)u_{N-1}'(s)+(s-1)u_{N-1}''(s)-3su_{N-1}''(s)+2K_{N}(s) \label{solvekn}\\
&=-5u_{N-1}'(s)-3su_{N-1}''(s)-\xi_{N}(s)\ , \label{xieqn}
\end{align}
where 
\begin{equation}
\xi_{N}(s) = \int_{0}^{\infty}e^{sx}x^{2}x^{\alpha}e^{-x}(L_{N-1}(x)L'_{N-2}(x)-L_{N-1}'(x)L_{N-2}(x))\,\mathrm{d}x\ . \label{xi}
\end{equation}
The difference of Laguerre polynomials in \eqref{xi} is nothing but the Christoffel-Darboux form of the eigenvalue density $\rho^{(2)}_{N-1}(x)$, cf. formula \eqref{beta2dens}. We deduce that
\begin{equation}
\xi_{N}(s) = \frac{1}{4d_{N}}M^{(2)\hspace{1pt}\prime}_{N-1}(s)\ ,
\end{equation}
and hence $K_{N}$ can be expressed in terms of explicitly known quantities. Using \eqref{eq:u_n1} 
and solving \eqref{solvekn}-\eqref{xieqn} for $2K_{N}$, we insert the results into \eqref{kneqn} and obtain the closed equation
\begin{equation}
(s^{2}-1/4)Y''+\left(4s+\frac{\alpha-1}{2}+N\right)Y'+\frac{9-\alpha^{2}}{4}Y+\frac{1}{4d_{N}}(3s^{2}M_{N-1}^{(2)\prime\prime}+(11s-1)M_{N-1}^{(2)\prime}+5M_{N-1}^{(2)})=0\ . \label{t1eqn}
\end{equation}
All that remains is to rewrite this in terms of $M_{N}^{(1)}(s)$ and its derivatives using \eqref{eq:M1M2}:
\begin{equation}
\begin{split}
&(4s^{2}-1)M_{N}^{(1)\prime\prime}+(16s+2(\alpha-1)+4N)M_{N}^{(1)\prime}+(9-\alpha^{2})M_{N}^{(1)}\\
&+(1-s^{2})M_{N-1}^{(2)\prime\prime}-(5s+2\alpha-1+4N)M_{N-1}^{(2)\prime}+(\alpha^{2}-4)M_{N-1}^{(2)}=0\ . \label{QMcoupled}
\end{split}
\end{equation}
The result \eqref{odeb1} now follows after using \eqref{ode} of Theorem \ref{thm:Haagerup} to eliminate the variable $M_{N-1}^{(2)\prime\prime}(s)$ appearing in \eqref{QMcoupled}.
\end{proof}
\section{Conclusions and remarks}
We considered the average of power traces $\E\left[\Tr Q^k\right]$ of the time-delay matrix for ballistic chaotic cavities. The sample-to-sample average can be computed using a Random Matrix Theory ansatz.  The large-$N$ expansion of these averages have been computed (recursively) for systems with and without broken time reversal symmetry ($\beta=2$ and $\beta=1$, respectively); we suggest that the coefficients of the large-$N$ expansion are whole numbers (Conjecture~\ref{conj:1}), thus extending a previous conjecture for the leading order of higher cumulants. A heuristic explanation of the conjecture comes from the semiclassical approach to chaotic scattering. 

We conclude with a last  remark. We recall that, up to a scaling factor, the moments of the time-delay matrix are statistically identical to the inverse moments $\Tr (W_N^{-k})$  of a specific Wishart ensemble. More precisely, for $\beta=2$ for concreteness, $W_N=XX^{\dagger}$ where $X=(x_{ij})$ is a $N\times 2N$ matrix with independent standard complex Gaussian entries. 
It is known that the large-$N$ expansion of moments of Gaussian (GUE) matrices are integers and they are related to a very precise enumeration problem.  Similar enumeration problems emerge for the  moments of Wishart matrices. The salient feature of Gaussian and Wishart matrices, and the one which is fundamental for
applications to graphical enumeration, is that expectations of general polynomial functions
can be reduced to a counting of Wick pairings. When considering the average of traces of powers of \emph{inverse} Wishart matrices (our situation), Wick's calculus no longer applies, but the integer nature of the coefficients nevertheless suggests an underlying combinatorial structure yet to be unveiled. 
\begin{acknowledgments}
FDC and FM acknowledge  support  from EPSRC Grant No.\ EP/L010305/1. FDC acknowledges partial support from the Italian National Group of Mathematical Physics (GNFM-INdAM). NS wishes to acknowledge support of a
Leverhulme Trust Early Career Fellowship (ECF-2014-309). PV acknowledges the
stimulating  research  environment  provided  by  the  EPSRC  Centre  for  Doctoral
Training in Cross-Disciplinary Approaches to Non-Equilibrium Systems (CANES,
EP/L015854/1). FDC wishes to thank Margherita Disertori for stimulating discussions. The authors are grateful to a referee for useful comments. This paper has no underlying data.
\end{acknowledgments}

\appendix
\section{Numerical tables}
\label{app:num}
Here are a few values of the moments $\tau_k^{(\beta)}=N^{k-1}\E[ \Tr Q^k]$. Recall that $\tau_0^{(\beta)}=\tau_1^{(\beta)}=1$. For $\beta=2$:
\barr
\tau_2^{(2)}&=&\frac{2 N^2}{N^2-1},\nonumber\\
\tau_3^{(2)}&=&\frac{6 N^4}{(N^2-4)(N^2-1)},\nonumber\\
\tau_4^{(2)}&=&\frac{22 N^6+2N^4}{(N^2-9)(N^2-4)(N^2-1)},\nonumber\\
\tau_5^{(2)}&=&\frac{90 N^8+30N^6}{(N^2-16)(N^2-9)(N^2-4)(N^2-1)},\nonumber\\
\tau_6^{(2)}&=&\frac{394 N^{10}+310 N^8+16 N^6}{(N^2-25)(N^2-16)(N^2-9)(N^2-4)(N^2-1)}.\nonumber\\
\earr
For $\beta=1$:
\barr
\tau_2^{(1)}&=&\frac{2 N^2}{(N-2) (N+1)},\nonumber\\
\tau_3^{(1)}&=&\frac{6 N^4}{(N-4) (N-2) (N+1) (N+2)},\nonumber\\
\tau_4^{(1)}&=&\frac{22 N^6-4 N^5}{(N-6) (N-4) (N-2) (N+1) (N+2) (N+3)},\nonumber\\
\tau_5^{(1)}&=&\frac{90 N^8-60 N^7}{(N-8) (N-6) (N-4) (N-2) (N+1) (N+2) (N+3) (N+4)},\nonumber\\
\tau_6^{(1)}&=&\frac{394 N^{10}-998 N^9-48 N^8-184 N^7-64 N^6}{(N-10) (N-8) (N-6) (N-4) (N-3) (N+1) (N+2) (N+3) (N+4) (N+5)}.\nonumber\\
\earr

\begin{table}[h]
\begin{tabular}{l*{8}{l}}
\toprule
$\boxed{\tau_{k,g}^{(2)}}$ & \multicolumn{7}{c}{$g$} \\
$k\qquad$ & 0 & 1 & 2 & 3 & 4 & 5 & 6\\
\cmidrule(rr){2-8}
 0 & $1$ & $0$ &  $0$ &  $0$ &  $0$ & $0$ & $0$\\
 1 & $1$ &  $0$ &   $0$ & $0$ &  $0$ &  $0$ & $0$ \\
2 &   $2$ & $0$ & $2$ &  $0$ & $2$ &  $0$ & $2$  \\
3 & $6$ &  $0$ & $30$ &  $0$& $126$ & $0$ & $510$  \\
4 &  $22$ &  $0$ & $310$ & $0$ & $3262$ &  $0$  & $31270$ \\
5 &  $90$ & $0$ &$ 2730$ &  $0$& $57330$ & $0$ & $1048410$\\
6 & $394$ &  $0$ &$ 21980$ &  $0$ & $805854$ &  $0$ & $24848560$ \\
7 & $1806$ &  $0$ & $167076$ &  $0$& $9781002$ &  $0$ & $468660192$ \\
8 & $8558$ &  $0$ &$ 1220100$ &  $0$ & $106963626$ &  $0$ & $7510405760$ \\
\toprule
\toprule
$\boxed{\tau_{k,g}^{(1)}}$& \multicolumn{7}{c}{$g$} \\
$k\qquad$ & 0 & 1 & 2 & 3 & 4 & 5 & 6\\
\cmidrule(rr){2-8} 
0 & $1$ & $0$ & $0$ & $0$ & $0$ & $0$ & $0$  \\
1 & $1$ & $0$ & $0$ & $0$ & $0$ & $0$ & $0$ \\
2 & $2$ & $2$ & $6$ & $10$ & $22$ & $42$ & $86$ \\
3 & $6$ & $18$ & $102$ & $378$ & $1638$ & $6426$ & $26214$ \\
4 & $22$ & $128$ & $1142$ & $7048$ & $47454$ & $291696$ & $1821094$ \\
5 & $90$ & $840$ & $10650$ & $96000$ & $904530$ & $7786680$ & $66945450$ \\
6 & $394$ & $5306$ & $89576$ & $1092460$ & $13529862$ & $152881422$ & $1704027412$\\
7 & $1806$ & $32802$ & $705012$ & $11060700$ & $172576362$ & $2451889734$ & $34038711504$\\
8 & $8558$ & $200064$ & $5297924$ & $103150528$ & $1966038698$ & $34052988736$ & $572050771840$\\
\bottomrule
\end{tabular}
\caption{A few values of $\tau_{k,g}^{(2)}$ (top) and $\tau_{k,g}^{(1)}$ (bottom) computed from the recursions~\eqref{eq:recoeff} and ~\eqref{eq:recoeffb1}-\eqref{eq:recoeffb1bis}. }
\label{tab:I}
\end{table}

\section{The generating functions}
\label{app:num}
The first few polynomials $R_g(z)$ defined in~\eqref{eq:R_g} are:
\begin{align*}
R_{2}(z) &= 2 z^2,\\
R_{4}(z) &= 16 z^6-24 z^5+6 z^4+60 z^3+2 z^2,\\
R_{6}(z) &= 360 z^{10}-96 z^9-304 z^8+4464 z^7-14110 z^6\\
         &\quad +12600 z^5+7572 z^4+408z^3+2 z^2,\\
R_{8}(z) &= 16128 z^{14}+46656 z^{13}-64776 z^{12}+413136 z^{11}-2210472 z^{10}\\
         &\quad +5724216z^9-5754378 z^8-3243996 z^7+7652766 z^6+2426400 z^5\\
         &\quad +152298 z^4+1908z^3+2 z^2,\\
R_{10}(z)   &= 1209600 z^{18}+9106560 z^{17}+936576 z^{16}+34986528 z^{15}-351879792 z^{14}\\
         &\quad +1396450368 z^{13}-2988047424 z^{12}+2088897408 z^{11}+5092739154 z^{10}\\
         &\quad -11252766096 z^9+2587036584 z^8+6426673488 z^7+1479326572 z^6\\
         &\quad +98620176 z^5+1927176 z^4+8016 z^3+2 z^2.
\end{align*}
The first few polynomials $P_k(\zeta)$ of Corollary~\ref{cor:rec2simp} have the following simple expression:
\begin{align*}
P_{2}(\zeta) &= 2, &P_{3}(\zeta) &= 6,\\
P_{4}(\zeta) &= 2\zeta+22, &P_{5}(\zeta) &= 30\zeta+90,\\
P_{6}(\zeta) &= 16\zeta^2 +310 \zeta +394, &P_{7}(\zeta) &=504 \zeta ^2+2730 \zeta+1806,\\
P_{8}(\zeta) &= 360 \zeta ^3+9422 \zeta ^2+21980 \zeta +8558, &P_{9}(\zeta) &= 18264 \zeta ^3+135954 \zeta ^2+167076 \zeta +41586.
\end{align*}
Here are the first few generating functions $F^{(1)}_g(z)$ computed from the recursion of Corollary \ref{cor:rec2simp1b1}:
\begin{align*}
F^{(1)}_0(z) &=\frac{3-z}{2}-\frac{y(z)^{1/2}}{2},\\
F^{(1)}_1(z) &=\frac{1-3 z}{2 y(z)}-\frac{1}{2 y(z)^{1/2}},\\
F^{(1)}_2(z) &=\frac{z^2-3 z}{y(z)^2}+\frac{3 z^3-4 z^2+3 z}{y(z)^{5/2}},\\
F^{(1)}_3(z) &=-\frac{2 z}{y(z)^{7/2}} \left(2 z^3-9 z^2+19 z+3\right)-\frac{2 z}{y(z)^4} \left(6 z^4-5 z^3+9 z^2-15 z-3\right),\\
F^{(1)}_4(z) &=\frac{2}{y(z)^{11/2}} \left(36 z^7+20 z^6+24 z^5-219 z^4+216 z^3+163 z^2+6 z\right)\\
&+\frac{2}{y(z)^6} \left(12 z^8-132 z^7+618 z^6-1830 z^5+1840 z^4+720 z^3-134 z^2-6 z\right),\\
F^{(1)}_5(z) &=-\frac{2 z}{y(z)^{13/2}} \left(96 z^7-456 z^6+2992 z^5-7068 z^4+3089 z^3+8214 z^2+979 z+12\right)\\
&-\frac{2 z}{y(z)^7} \left(288 z^8+776 z^7-336 z^6-2916 z^5+6276 z^4-1312 z^3-7560 z^2-964 z-12\right),\\
F^{(1)}_6(z) &=\frac{2}{{y(z)^{17/2}}} (2880 z^{11}+15588 z^{10}-3552 z^9-53360 z^8+139938 z^7\\
&-121877 z^6-186156 z^5+329334 z^4+100650 z^3+4699 z^2+24 z)\\
&+\frac{2}{y(z)^9}(960 z^{12}-9504 z^{11}+62944 z^{10}-373248 z^9+997768 z^8-981480 z^7\\
&-1012248 z^6+2243256 z^5+250584 z^4-75672 z^3-4584 z^2-24 z).
\end{align*}

\end{document}